\documentclass[twocolumn,footinbib]{revtex4}

\usepackage{amsmath}
\usepackage{amssymb}
\usepackage{url}
\usepackage{graphicx}

\usepackage{amsthm}
\usepackage{enumerate}
\newtheorem{theorem}{Theorem}[section]
\newtheorem{lemma}[theorem]{Lemma}

\newtheorem{remark}[theorem]{Remark}

\theoremstyle{plain}

\newcommand*{\cF}{\mathcal{F}}

\newcommand*{\cH}{\mathcal{H}}

\newcommand*{\cM}{\mathcal{M}}

\newcommand*{\cS}{\mathcal{S}}
\newcommand*{\cU}{\mathcal{U}}

\newcommand*{\cV}{\mathcal{V}}

\newcommand*{\cX}{\mathcal{X}}

\newcommand*{\cZ}{\mathcal{Z}}

\newcommand*{\cHb}{\mathcal{\bar{H}}}
\newcommand*{\cZb}{\mathcal{\bar{Z}}}

\DeclareMathOperator*{\Sym}{Sym}

\newcommand*{\eps}{\varepsilon}
\newcommand*{\tr}{\mathsf{tr}}
\newcommand*{\id}{\openone}

\newcommand*{\ket}[1]{|#1\rangle}
\newcommand*{\bra}[1]{\langle #1|}
\newcommand*{\proj}[1]{\ket{#1}\bra{#1}}
\newcommand*{\spr}[2]{\langle #1|#2 \rangle}

\newcommand*{\states}{\cS}

\newcommand*{\Expect}{\mathrm{E}}

\newcounter{fnnumber}

\begin{document}

\title{A de Finetti representation theorem for infinite dimensional
  quantum systems and applications to quantum cryptography}


\author{R. Renner} \affiliation{Institute for Theoretical Physics, ETH
  Zurich, CH-8093 Zurich, Switzerland.}

\author{J.I. Cirac}
\affiliation{Max-Planck-Institut f\"ur Quantenoptik,
Hans-Kopfermann-Str.~1, D-85748 Garching, Germany.}

\begin{abstract}
  According to the quantum de Finetti theorem, if the state of an
  $N$-partite system is invariant under permutations of the subsystems
  then it can be approximated by a state where almost all subsystems
  are identical copies of each other, provided $N$ is sufficiently
  large compared to the dimension of the subsystems. The de Finetti
  theorem has various applications in physics and information theory,
  where it is for instance used to prove the security of quantum
  cryptographic schemes. Here, we extend de Finetti's theorem, showing
  that the approximation also holds for infinite dimensional systems,
  as long as the state satisfies certain experimentally verifiable
  conditions. This is relevant for applications such as quantum key
  distribution (QKD), where it is often hard|or even impossible|to
  bound the dimension of the information carriers (which may be
  corrupted by an adversary). In particular, our result can be applied
  to prove the security of QKD based on weak coherent states or
  Gaussian states against general attacks.
\end{abstract}


\date{September 12, 2008}

\maketitle

\section{Introduction}

Systems studied in physics often consist of a large number of
identical subsystems. Examples include any type of matter with the
individual molecules as subsystems, or a light field consisting of
many modes. Similarly, in the context of quantum information
processing, one typically considers settings involving a large number
of identical information carriers, such as the photons sent over an
optical fiber.  In all these cases, the state of the overall system is
described by a density operator on a product space $\cH^{\otimes N}$.

A main difficulty when studying large composite systems is that their
dimension, and hence the number of parameters needed to describe their
state, grows exponentially in the number $N$ of subsystems. This is
particularly problematic if one wants to prove that a certain
statement holds for all possible states of the system. In the context
of quantum information processing, the necessity of such proofs
arises, for instance, when analyzing the security of cryptographic
protocols.  Here, an adversary may maliciously manipulate the
information carriers, and security must be guaranteed for any
resulting state.

The analysis of large composite quantum systems can be vastly
simplified under certain symmetry assumptions, using a quantum version
of de Finetti's classical representation theorem~\cite{deFine37}
proposed recently in~\cite{Renner05,Renner07}. The theorem states that
multi-partite density operators which are invariant under permutations
of the subsystems are approximated by convex combinations of density
operators which have \emph{i.i.d.\ structure} $\sigma^{\otimes N}$ on
most subsystems~\footnote{The term \emph{i.i.d.}\ (for
  \emph{independent and identically distributed}) is traditionally
  used in probability and information theory for random variables
  $X_1, \ldots, X_N$ whose probability mass (or density) function is
  of the form~$P_{X_1 \cdots X_N} = P_X \times \cdots \times P_X$. We
  use it here for multi-partite quantum states of the form
  $\sigma^{\otimes N}$.}. I.i.d.\ states can be easily parametrized
(they are characterized by the state $\sigma$ of a single subsystem),
and a huge variety of tools are available to handle them, particularly
in the area of information theory~\cite{NieChu00}.

In information-theoretic applications, permutation symmetry of the
states can often be assumed to hold without loss of generality due to
inherent symmetries of the underlying problem or the processing
scheme. An important example, which we are going to study in more
detail, is \emph{quantum key distribution
  (QKD)}~\cite{BenBra84,Ekert91}. Roughly speaking, QKD is the art of
establishing a secret key between two distant parties, traditionally
called \emph{Alice} and \emph{Bob}, connected only by an insecure
quantum channel~\footnote{In addition, Alice and Bob need to be able
  to exchange classical messages authentically or, alternatively,
  share a short initial key.}.  Most QKD protocols have the property
that $N$ signals are exchanged sequentially between Alice and Bob, but
the order in which they are transmitted is irrelevant (as long as
Alice and Bob coordinate their communication). One can thus
equivalently assume that Alice and Bob reorder the signals according
to a randomly chosen permutation~\footnote{\label{ftn:perm} We
  emphasize here that this random permutation is only used in the
  theoretical analysis, but need not be implemented in the actual
  protocol, as shown
  in~\cite{ScaRen08}.}. \setcounter{fnnumber}{\thefootnote}
Consequently, even if an adversary manipulates the signals in an
arbitrarily malicious way, the $N$-partite density operator describing
Alice and Bob's information is permutation invariant.

The quantum de Finetti theorem now implies that, for assessing the
security of a QKD protocol, it is sufficient to consider the special
case where the state held by Alice and Bob (after communication over
the insecure channel) has i.i.d.\ structure.  This, however, exactly
corresponds to the situation arising in a \emph{collective
  attack}~\cite{BihMor97a,BihMor97b}, where the adversary is bound to
manipulate each of the transmitted signals independently and
identically. For a large class of protocols, security against
collective attacks is well understood and explicit formulas for the
key rate are known (see, e.g., \cite{DevWin05} for the rate of key
distillation protocols with one-way communication).

The reduction of security proofs to the special case of collective
attacks, however, only works for QKD schemes that use low-dimensional
signals.  This is because the de Finetti representation for states on
product spaces $\cH^{\otimes N}$ is subject to the constraint that the
dimension $d$ of the subsystems $\cH$ be sufficiently smaller than the
number $N$ of subsystems. In particular, the de Finetti representation
generally fails if $\cH$ is infinite-dimensional. (There exist
explicit examples of permutation invariant states $\rho^N$ on
$\cH^{\otimes N}$, with $\dim(\cH) = N$, such that any reduced state
$\rho^{k}$ on $\cH^{\otimes k}$, for $k \geq 2$, is highly entangled
and, hence, cannot be approximated by a convex combination of i.i.d.\
states~\cite{CKMR06}.)


Here, we show that the restriction of the de Finetti representation to
low-dimensional spaces $\cH$ can be circumvented under certain
experimentally verifiable conditions. More precisely, we prove that
for any permutation invariant state on a (possibly
infinite-dimensional) system $\cH^{\otimes N}$, the reduced state on
$\cH^{\otimes N'}$, for some $N' \approx N$, is approximated by a
mixture of density operators with i.i.d.\ structure, provided that the
outcomes of a measurement applied to a few subsystems lie within a
given range. As a specific example, we consider measurements with
respect to two canonical observables $X$ and $Y$ on $\cH = {\rm
  L}^2(\mathbb{R})$. The criterion then is that the outcomes of both
the $X$ and the $Y$ measurements have small absolute value.

In practical applications, this criterion is often easily
verifiable. For example, in \emph{continuous variable quantum
  cryptography}~\cite{Ralph99,Hiller00,Reid00,GotPre01,CeLeVa01,SiKoLe02,GroGra02,GAWBCG03,WLBSRL04},
which uses signals in $\cH = {\rm L}^2(\mathbb{R})$, measurements with
respect to two canonical observables $X$ and $Y$ are usually already
part of the protocol. Our extended version of de Finetti's theorem
then implies that these protocols are secure against the most general
attacks, provided they are secure against collective attacks. The
latter type of security is already proved for many practical
continuous variable schemes~(see, e.g., \cite{GarCer06}, which is
based on~\cite{WoGiCi06}, and~\cite{NaGrAc06}).

\bigskip

The remainder of this paper is organized as follows. After introducing
some notation and terminology, we start in Section~\ref{sec:tech} with
the proof of the technical lemmas and theorems. These are the building
blocks for the derivation of our main claim that permutation invariant
states are approximated by almost i.i.d.\ states, as described in
Section~\ref{sec:combination}. (For a first reading, one may skip
Section~\ref{sec:tech} and directly start with
Section~\ref{sec:combination}, where it is shown how the individual
technical claims are combined.)  Finally, we discuss how our result
can be applied to prove the security of QKD schemes
(Section~\ref{sec:application}).

\section{Notation and definitions} \label{sec:notation}

\subsection{Symmetry and permutation invariance}

Let $S_n$ be the set of permutations on $\{1, \ldots, n\}$ and let
$\cH$ be a Hilbert space. The \emph{symmetric subspace of
  $\cH^{\otimes n}$}, denoted $\Sym^n(\cH)$, consists of all vectors
$\Phi \in \cH^{\otimes n}$ such that $\pi \Phi = \Phi$ for all $\pi
\in S_n$. The projector on $\Sym^n(\cH)$ can be written as
\begin{align} \label{eq:symproj}
  P_{\Sym^n(\cH)} = \frac{1}{n!} \sum_{\pi \in S_n} \pi \ .
\end{align}

We denote by $\states(\cH)$ the set of density operators on the
Hilbert space $\cH$. An operator $\rho^{n} \in \states(\cH^{\otimes
  n})$ is said to be \emph{permutation invariant} if $\pi \rho^{n}
\pi^\dagger = \rho^{n}$ for all permutations~$\pi$.

\subsection{Restricted symmetric subspaces}

Let $\cHb$ be a subspace of $\cH$ and let $k, n \in \mathbb{N}$.  We
define $P^{k+n}_{\cHb^{\otimes n}}$ as the projector onto the subspace
of $\cH^{\otimes k+n}$ spanned by all vectors in $\pi(\cH^{\otimes k}
\otimes \cHb^{\otimes n})$, for any $\pi \in S_{k+n}$. The projector
$P^{k+n}_{\cHb^{\otimes n}}$ can be decomposed into projectors $P_0 =
P_{\cHb}$ and $P_1 = P_{\cHb^\perp}$ onto $\cHb$ and its orthogonal
subspace $\cHb^\perp$, respectively,
\begin{align} \label{eq:Pkndef}
  P^{k+n}_{\cHb^{\otimes n}} = \sum_{\substack{\mathbf{b} \in \{0,1\}^{k+n}\\f_{\mathbf{b}} \leq \frac{k}{k+n}}} P_{b_1} \otimes \cdots \otimes P_{b_{k+n}} \ ,
\end{align}
where the sum ranges over all bitstrings $\mathbf{b} = (b_1, \ldots,
b_{k+n}) \in \{0,1\}^{k+n}$ whose \emph{relative frequency of $1$s},
\begin{align} \label{eq:relfreq}
  f_{\mathbf{b}} := \frac{1}{k} \sum_{k} b_i \ ,
\end{align}
is not larger than $\frac{k}{k+n}$.

Because $P^{k+n}_{\cHb^{\otimes n}}$ is permutation invariant it
commutes with any $\pi \in S_{k+n}$ and, hence, also with the
projector $P_{\Sym^{k+n}(\cH)}$ onto the symmetric subspace of
$\cH^{\otimes k+n}$ (see~\eqref{eq:symproj}). This implies that the
product $P^{k+n}_{\cHb^{\otimes n}} P_{\Sym^{k+n}(\cH)}$ is a
projector.  In the following, we denote by $\Sym^{k+n}(\cH,
\cHb^{\otimes n})$ the support of this projector. The space
$\Sym^{k+n}(\cH, \cHb^{\otimes n})$ thus consists of all symmetric
vectors that can be written as superpositions of vectors of the form
$\pi(\Phi \otimes \bar{\Phi})$, for some $\Phi \in \Sym^k(\cH)$,
$\bar{\Phi} \in \Sym^n(\cHb)$, and $\pi \in S_{k+n}$.

In the special case where $\cHb = \mathrm{span}\{\nu\}$ is the vector
space spanned by a single vector $\nu \in \cH$, we also write
$\Sym^{k+n}(\cH, \nu^{\otimes n})$ instead of $\Sym^{k+n}(\cH,
\mathrm{span}\{\nu\}^{\otimes n})$ and call its elements
\emph{$\binom{k+n}{n}$-i.i.d.\ vectors (along $\nu$)}. We also say
that a density operator $\rho^{k+n}$ is \emph{almost i.i.d.}\ if its
support is contained in $\Sym^{k+n}(\cH, \nu^{\otimes n})$, for some
$k \ll n$.

\subsection{Measurements}

Let $U$ and $V$ be nonnegative operators on a Hilbert space $\cH$
satisfying $U \leq \id$ and $V \leq \id$. We define the function
$\gamma_{U \to V}$ on $[0,1]$ by
\begin{align} \label{eq:gammadef}
  \gamma_{U \to V}(\delta) := \sup \{\tr(V \sigma) : \, \sigma \in \states(\cH); \, \tr(U \sigma) \leq \delta\} \ .
\end{align}
If $U$ and $V$ are POVM elements then $\gamma_{U \to V}(\delta)$
corresponds to the maximum probability of obtaining outcome $V$ when
measuring a state $\sigma$ for which the probability of outcome~$U$ is
at most~$\delta$.

\section{Technical statements} \label{sec:tech}

\subsection{Measurement statistics}

Let $\cU = \{U_0, U_1\}$ and $\cV = \{V_0, V_1\}$ be two binary POVMs
on $\cH$ with the property that $\gamma_{U_1 \to V_1}(\delta)$ is
small for small $\delta$. In other words, for any state $\sigma$,
outcome~$1$ of measurement $\cV$ has small probability whenever
outcome~$1$ of measurement $\cU$ has small probability. Intuitively,
we would then expect that the following holds. If $k$ subsystems of a
$(k+n)$-partite permutation invariant state are measured according to
$\cU$, resulting in a low number of outcomes~$1$, then the number of
outcomes~$1$ when measuring the $n$ remaining subsystems according to
$\cV$ is small, too. The following lemma makes this intuition more
precise.

\begin{lemma} \label{lem:probbound} Let $\cU = \{U_0, U_1\}$ and $\cV
  = \{V_0, V_1\}$ be POVMs on $\cH$, let $n \geq 2 k$, and let $(X_1,
  \ldots, X_{k+n})$ be the $(k+n)$-partite classical outcome of the
  measurement $\cU^{\otimes k} \otimes \cV^{\otimes n}$ applied to any
  permutation invariant $\rho^{k+n} \in \states(\cH^{\otimes
    k+n})$. Then, for any $\delta > 0$,
  \begin{align*}
    \Pr\bigl[f_{X_{k+1} \cdots X_{k+n}} > \gamma_{U_1 \to V_1}(f_{X_{1} \cdots X_{k}} + \delta) + \delta \bigr] \leq 8 k^{\frac{3}{2}} e^{-k\delta^2} \ ,
  \end{align*}
  where $f_{\mathbf{X}}$ denotes the relative frequency of $1$s in
  $\mathbf{X}$ (see~\eqref{eq:relfreq}).
\end{lemma}

Qualitatively, the statement of Lemma~\ref{lem:probbound} is a special
case of Lemma~4.1 of~\cite{ChReEk04}. For completeness, we give a
proof in the appendix, which also yields tighter bounds for the choice
of parameters we are interested in.

\subsection{Bounding the probability of projecting into a
  low-dimensional subspace}

In this section, we derive a bound on the quantity $\gamma_{U_1 \to
  V_1}$ for the case where $V_1$ corresponds to the predicate that a
measurement of $X^2 + Y^2$, for two canonical observables $X$ and $Y$
on ${\cal H}={\rm L}^2(\mathbb{R})$, is larger than a threshold $n_0$,
and where $U_1$ is the predicate that the outcome of a measurement
with respect to either $X^2$ or $Y^2$ is at least $\frac{n_0}{2}$.

For any Hermitian operator $Z$ and $z_0\in \mathbb{R}$ we define
$P^{Z\ge z_0}$ as the projector onto the subspace spanned by the
eigenspaces of $Z$ corresponding to (generalized) eigenvalues $z\ge
z_0$.

\begin{lemma} \label{lem:Nconstraint}
Let $X$ and $Y$ be two canonical operators ($[X,Y]=i$),
$n_0$ a positive integer, and define
\[
U_1:= \frac{1}{2}P^{X^2\ge n_0/2} + \frac{1}{2}P^{Y^2\ge n_0/2}
\; {\it and} \;
V_1:= P^{X^2+Y^2\ge n_0+1}
\]
Then $\gamma_{U_1\to V_1}(\delta) \le 4 \delta + \frac{4}{c_0\sqrt{\pi
    n_0}}e^{-n_0c_0^2}$, with $c_0=1-\frac{1}{\sqrt{2}}$.
\end{lemma}

\begin{proof}
  The proof consists of several steps. First, we define an operator
  $W_1$ and show that $V_1\le 2W_1$. Then we show that, up to a
  constant, $W_1$ is upper bounded by $2U_1$.

Let us start by defining
\begin{align*}
W_1:= \frac{1}{\pi}\int d\mu_\alpha \; |\alpha\rangle\langle\alpha|
\end{align*}
where $|\alpha\rangle$ denotes a coherent state and the integral is
extended to the complex plane with $|\alpha|^2\ge n_0$. By expanding
$W_1$ in the Fock basis, $\{|n\rangle_f\}_{n=0}^\infty$, one obtains
that $W_1=\sum q_n|n\rangle_f\langle n|$ with
$q_n=\Gamma(n+1,n_0)/\Gamma(n+1,0)$, where $\Gamma$ is the incomplete
Gamma function~\cite{Grad00}. Since $q_{n+1}\ge q_n>0$, we can write
$V_1\le q_{n_0}^{-1} W_1$, where
$q_{n_0}^{-1}=\Gamma(n_0+1,0)/\Gamma(n_0+1,n_0)<2$, which concludes
the first part.

For the second part, we first extend our Hilbert space to ${\cal
H}_1\otimes {\cal H}_2$, and show that we can write
\begin{equation}
\label{W1}
W_1 = \int d x d y \;_f\langle
0|U(|x\rangle_X\langle x|\otimes|y\rangle_Y\langle
y|)U^\dagger|0\rangle_f,
\end{equation}
where the integral is defined for $x,y\in \mathbb{R}$ with the
restriction $x^2+y^2\ge n_0$. Here $|0\rangle_f\in {\cal H}_2$, and
$|x\rangle_{X,Y}$ denote generalized eigenstates of $X$ and $Y$,
respectively. Furthermore, $U=e^{\frac{\pi}{4}(a_1\otimes
  a_2^\dagger-a_1^\dagger\otimes a_2)}$ is the so--called beam
splitter operator~\cite{BeamSplit}, where
$a_{1,2}:=(X_{1,2}+i Y_{1,2})/\sqrt{2}$ are the annihilation operators
acting on the first and second system, respectively.  This expression
for $W_1$ can be derived by showing that $|f_{x,y}\rangle:=_f\langle
0|U|x\rangle_X\otimes|y\rangle_Y=\pi^{-1/2}|\alpha\rangle$, with
$\alpha=x+iy$. This, in turn, can be proved by realizing that it is an
eigenstate of the annihilation operator,
\begin{eqnarray}
&&a_1 |f_{x,y}\rangle = _f\langle
0|(a_1+a_2^\dagger) U|x\rangle_X\otimes|y\rangle_Y\nonumber\\
&=&_f\langle
0|U(X_1+i Y_2)|x\rangle_X\otimes|y\rangle_Y
=(x+i y)|f_{x,y}\rangle,\nonumber
\end{eqnarray}
where we have used the fact that $_f\langle 0|a_2^\dagger=0$ and that
$U^\dagger (a_1+a_2^\dagger)U=X_1+i Y_2$. The normalization factor can
be obtained by noting that the integral over the complex plane of
$|\alpha\rangle\langle\alpha|=\pi \id$. By looking at the integration
domain in (\ref{W1}) it is clear that $W_1\le A+B$, where
\begin{eqnarray}
A&=&\int d x d x' \;
_f\langle 0|U(|x\rangle_X\langle x|\otimes
|x'\rangle_X\langle x'|)U^\dagger|0\rangle_f,\nonumber\\
B&=&\int d x d x' \;
_f\langle 0|U(|x'\rangle_Y\langle x'|\otimes
|x\rangle_Y\langle x|)U^\dagger|0\rangle_f,\nonumber
\end{eqnarray}
where the integral is restricted to $|x|^2\ge n_0/2$ and
$-\infty<x<\infty$, and we have used that the integral of
$|x'\rangle_X\langle x'|$ is equal to that of $|x'\rangle_Y\langle
x'|$. Using
$U|x\rangle_X\otimes|x'\rangle_X=|(x+x')/\sqrt{2}\rangle_X\otimes
|(x-x')/\sqrt{2}\rangle_X$, $|_f\langle
0|x\rangle_X|^2=e^{-x^2}/\sqrt{\pi}$, and changing variables in
the integrals ([$z'=(x+x')/\sqrt{2}$, $z=\sqrt{2}x$] we obtain
\[
A=\frac{1}{\sqrt{\pi}}\int_{|z|^2\ge n_0}dz\; e^{-(z-X)^2}=:F(X).
\]
Analogously, $B=F(Y)$. It is straightforward to show that for all
$a>0$, $F(X)\le P^{X^2\ge a^2}+F(a)$, and similarly for $F(Y)$.
Noting that $F(a)\le
(1/\sqrt{\pi})e^{-(\sqrt{n_0}-a)^2}/(\sqrt{n_0}-a)$, for $a \in [0,
\sqrt{n_0}]$, and choosing $a=\sqrt{n_0/2}$ we conclude the proof.
\end{proof}

Lemma~\ref{lem:lowdimbound} below is a corollary of
Lemma~\ref{lem:probbound} and Lemma~\ref{lem:Nconstraint}. It allows
to restrict the support of a $(2k+n)$-partite permutation invariant
density operator, provided that measurements of the two canonical
operators $X$ and $Y$ on $k$ subsystems only result in small values.

\begin{lemma} \label{lem:lowdimbound} Let $X$ and $Y$ be two canonical
  operators on $\cH$, let $n \geq 2 k$, let $\cHb$ be the support of
  $P^{X^2 + Y^2 \leq n_0}$, for any $n_0 \geq 12 \ln
  \frac{7(k+n)}{k}$, and let $\rho^{2 k + n}$ be a permutation
  invariant density operator on $\cH^{\otimes 2 k + n}$. Let $(Z_1,
  \ldots, Z_k)$ be the outcomes of measurements of $k$ subsystems of
  $\rho^{2 k +n}$ with respect to $X$ and $Y$ (each chosen with
  probability $\frac{1}{2}$) and let $\cF$ be the event that the
  projection $P^{k+n}_{\cHb^{\otimes n}}$ applied to the remaining
  ${k+n}$ subsystems fails. Then
  \begin{align*}
    \Pr\bigl[(\max_{i=1}^k Z_i^2 < \frac{n_0}{2}) \wedge \cF\bigr]
    \leq 8 k^{\frac{3}{2}} e^{-\frac{k^3}{49 (k+n)^2}} \ .
  \end{align*} 
\end{lemma}

\begin{proof}
  Let $U_1$ and $V_1$ be defined as in
  Lemma~\ref{lem:Nconstraint}. Furthermore, let $X_1, \ldots, X_{k+n}$
  be the outcomes of the POVM $\cU^{\otimes k} \otimes \cV^{\otimes
    n}$ defined by $\cU = \{\id-U_1, U_1\}$ and $\cV = \{\id-V_1,
  V_1\}$, as in Lemma~\ref{lem:probbound}. The probability we want to
  bound can then be rewritten as
  \begin{align*}
    \Pr\bigl[(\max_{i=1}^k Z_i^2 < \frac{n_0}{2}) \wedge \cF\bigr] 
  =
    \Pr\bigl[(f_{\mathbf{X}^k} = 0) \wedge (f_{\mathbf{X}^n} > {\textstyle \frac{k}{k+n}}) \bigr] \ .
  \end{align*}
  where $f_{\mathbf{X}^k}$ and $f_{\mathbf{X}^n}$ are the frequencies
  of $1$s in the tuples $\mathbf{X}^k = (X_1, \ldots, X_k)$ and
  $\mathbf{X}^n = (X_{k+1}, \ldots, X_{k+n})$, respectively.  With
  $\delta := \frac{k}{7(k + n)}$, we have
  \begin{align*}
  \gamma_{U_1 \to V_1}(\delta) + \delta \leq 5 \delta + {\textstyle \frac{8}{\sqrt{n_0}}} e^{-\frac{1}{12} n_0} \leq
  {\textstyle \frac{k}{k + n}} \ ,
  \end{align*}
  and, hence, the probability above can be bounded by
  \begin{multline*}
    \Pr\bigl[(f_{\mathbf{X}^k} = 0) \wedge (f_{\mathbf{X}^n} > {\textstyle \frac{k}{k+n}}) \bigr] \\
  \leq
    \Pr\bigl[(f_{\mathbf{X}^k} = 0) \wedge (f_{\mathbf{X}^n} > \gamma_{U_1 \to V_1}(f_{\mathbf{X}^k} + \delta) + \delta) \bigr] \\
  \leq
    \Pr\bigl[f_{\mathbf{X}^n} > \gamma_{U_1 \to V_1}(f_{\mathbf{X}^k} + \delta) + \delta\bigr] \ .
  \end{multline*}
  The claim then follows from Lemma~\ref{lem:probbound}.
\end{proof}

\begin{remark} \label{rem:lowdimbound} It is straightforward to
  generalize Lemma~\ref{lem:lowdimbound} to other measurements,
  specified by an arbitrary POVM $\cM = \{M_z\}_{z \in \cZ}$. The
  condition $\max_i Z_i^2 < \frac{n_0}{2}$ may then be replaced by the
  requirement that the outcomes $Z_i$ are contained in a certain set
  $\cZb \subseteq \cZ$ such that for any $\delta > 0$
    \begin{align*}
    \gamma_{U_1 \to P_{\cHb}^\perp}(\delta) \leq O(\delta) \ ,
  \end{align*}
  where $U_1 := \sum_{z \notin \cZb} M_z$ and where $P_{\cHb}^\perp$
  denotes the projection onto the subspace orthogonal to a
  finite-dimensional subspace $\cHb$, which may be chosen depending on
  $\delta$. For the considerations below
  (Section~\ref{sec:combination}), however, the dimension $d$ of
  $\cHb$ needs to be bounded by $d \leq O(\delta^{-\frac{3}{2}})$, so
  that $d \leq O((\frac{n}{k})^{\frac{3}{2}})$.
\end{remark}

\subsection{Purification in restricted symmetric subspaces}

The de Finetti type statements formulated in Section~\ref{sec:dF}
below apply to states on the symmetric subspace. The following lemma,
which is a generalization of Lemma~4.2.2 of~\cite{Renner05} (see
also~\cite{CKMR06}), allows to extend these statements to general
permutation invariant density operators.

\begin{lemma} \label{lem:pur} Let $\cHb$ be a subspace of $\cH$ and
  let $\rho^{2 k + n} \in \states(\cH^{\otimes 2k+n})$ be permutation
  invariant with support contained in the support of
  $P^{2k+n}_{\cHb^{\otimes k+n}}$. Then there exists a purification of
  $\rho_{2k+n}$ on $\Sym^{2k+n}({\cH \otimes \cH}, ({\cHb \otimes
    \cHb})^{\otimes n})$.
\end{lemma}

\begin{proof}
  Let $\{e_j\}_{j \in J}$ be an orthonormal basis of $\cH$ such that
  $\{e_j\}_{j \in K}$, for some $K \subset J$, is a basis of $\cHb$.
  We can then define a vector $\Phi \in (\cH \otimes \cH)^{\otimes
    2k+n}$ by
  \begin{align} \label{eq:Psidef} \Phi = \sum_{\mathbf{j} \in
      J^{2k+n}} (\rho^{2k+n} \otimes \id_{\cH}^{\otimes
      2k+n}) e_{\mathbf{j}} \otimes e_{\mathbf{j}}
  \end{align}
  where, for any $\mathbf{j} = (j_1, \ldots, j_{2k+n}) \in J^{2k+n}$,
  \begin{align*}
    e_{\mathbf{j}} = e_{j_1} \otimes \cdots \otimes e_{j_{2k+n}} \ .
  \end{align*}
  The state defined by $\Phi$ is obviously a purification of
  $\rho^{2k+n}$. Furthermore, because $\rho^{2k+n}$ is permutation
  invariant, we have for any $\pi \in S_{2 k + n}$
  \begin{align*}
    (\pi \otimes \pi) \Phi 
  & =
     (\pi \otimes \pi) \sum_{\mathbf{j} \in J^{2k+n}} (\rho^{2k+n} \otimes \id_{\cH}^{\otimes 2k+n}) e_{\mathbf{j}} \otimes e_{\mathbf{j}} \\
  & =
     \sum_{\mathbf{j} \in J^{2k+n}} (\rho^{2k+n} \otimes \id_{\cH}^{\otimes 2k+n}) (\pi e_{\mathbf{j}} \otimes \pi e_{\mathbf{j}}) = \Phi 
  \end{align*}
  and, hence, $\Phi \in \Sym^{2k+n}(\cH \otimes \cH)$. It thus remains
  to verify that $\Phi$ is an element of the support of
  $P^{2k+n}_{(\cHb \otimes \cHb)^{\otimes n}}$.

  Since $\rho^{2k + n}$ is contained in the support of
  $P^{2k+n}_{\cHb^{\otimes k+n}}$, the sum in~\eqref{eq:Psidef} can be
  restricted to terms such that $e_{\mathbf{j}}$ lies in the support
  of $P^{2k+n}_{\cHb^{\otimes k+n}}$, too (or, equivalently, the tuple
  $\mathbf{j}$ has at most $k$ entries outside $J$). This implies that
  $\Phi$ lies in the support of $P^{2 k + n}_{\cHb^{\otimes k+n}}
  \otimes P^{2 k + n}_{\cHb^{\otimes k+n}}$. The assertion then
  follows because this support is contained in the support of $P^{2 k
    + n}_{(\cHb \otimes \cHb)^{\otimes n}}$.
\end{proof}

\subsection{An extended de Finetti-type theorem} \label{sec:dF}

The purpose of this section is to derive a de Finetti-type theorem for
states on the symmetric subspace of product spaces with possibly
\emph{infinite-dimensional} subsystems (Theorem~\ref{thm:dFext}). We
start, however, with a de Finetti-type statement for \emph{finite}
dimensions (Lemma~\ref{lem:dF}). It can be seen as a strengthened
version of the exponential de Finetti theorem proposed
in~\cite{Renner05}. The claim is that any $(2k+n)$-partite symmetric
vector $\Phi$ is approximated by a \emph{superposition} of vectors
that are $\binom{k+n}{n}$-i.i.d.\ on $k+n$ subsystems. We note that,
in contrast, the approximation in~\cite{Renner05} has the form of a
\emph{convex combination} of $\binom{k+n}{n}$-i.i.d.\ states. A second
difference between Lemma~\ref{lem:dF} and the result
of~\cite{Renner05} is that we use the overlap (i.e., the scalar
product between vectors) instead of the trace distance to quantify the
quality of the approximation (this slightly simplifies the argument
below).

\begin{lemma} \label{lem:dF} Let $\cH$ be a $d$-dimensional Hilbert
  space and let $k, n \in \mathbb{N}$. There exists an isometry $U$
  from $\cH^{\otimes k}$ to a Hilbert space $\cH'$ with orthonormal
  basis $\{f_\nu\}_{\nu \in \cV}$, where $\cV$ is a finite set of unit
  vectors $\nu \in \cH$, such that the following holds.  For any unit
  vector $\Phi \in \Sym^{2k+n}(\cH)$ there exists a unit vector
  $\hat{\Phi} \in \cH' \otimes \Sym^{k+n}(\cH)$ of the form
  \begin{align} \label{eq:Phidef}
    \hat{\Phi} = {\textstyle \sqrt{\frac{1}{|\cV|}}} \sum_{\nu \in \cV} f_\nu \otimes \hat{\Phi}_\nu 
  \end{align}
  with $\hat{\Phi}_\nu \in \Sym^{k+n}(\cH, \nu^{\otimes n})$ such that 
  \begin{align} \label{eq:fidbound}
    \spr{\hat{\Phi}}{(U \otimes \id^{\otimes k+n})\Phi}  > 1- k^d e^{- \frac{k (k+1)}{2 k + n}} \ .
  \end{align}
\end{lemma}

Note that~\eqref{eq:fidbound} can be rewritten in terms of the
fidelity $F(\cdot, \cdot)$ as
\begin{align*}
  F(\hat{\Phi}, (U \otimes \id^{\otimes k+n})\Phi)  > 1- k^d e^{- \frac{k (k+1)}{2 k + n}} \ .
\end{align*}
The de Finetti theorem of~\cite{Renner05} (Theorem~4.3.2) can then be
obtained by taking the partial trace over $\cH'$ in both arguments of
$F(\cdot, \cdot)$ and converting the fidelity into a trace distance.

\begin{proof}
  The unitary group acts irreducibly on the subspace
  $\Sym^k(\cH)$. Hence, by Schur's lemma,
  \begin{align*}
    \int U^{\otimes k} (\proj{\nu_0})^{\otimes k} (U^\dagger)^{\otimes k} \omega(U)
  =
    {\textstyle \frac{1}{\dim(\Sym^k(\cH))}} P_{\Sym^k(\cH)}
  \end{align*}
  where $\nu_0$ is an arbitrary unit vector in $\cH$ and where
  $\omega$ is the Haar measure on the set of unitaries on $\cH$. Note
  that the integral on the left hand side can be approximated to any
  accuracy by a sum over a finite set $\cV$ of unit vectors $\nu \in
  \cH$. That is, for any $\mu > 0$ there exists a finite set $\cV$
  such that
  \begin{align*}
    \Bigl\|
      \frac{1}{|\cV|} \sum_{\nu \in \cV} (\proj{\nu})^{\otimes k} 
    -
      {\textstyle \frac{1}{\dim(\Sym^k(\cH))}} P_{\Sym^k(\cH)} 
    \Bigr\|
  \leq 
    \mu \ .
  \end{align*}

  Let $\cH'$ a Hilbert space with orthonormal basis $\{f_\nu\}_{\nu
    \in \cV}$ and define the linear map $\tilde{U}$ from $\cH^{\otimes
    k}$ to $\cH'$ by
  \begin{align*}
    \tilde{U} := {\textstyle \frac{\sqrt{\dim(\Sym^k(\cH))}}{\sqrt{|\cV|}}} \sum_{\nu \in \cV} \ket{f_\nu}\bra{\nu^{\otimes k}} \ .
  \end{align*}
  We then have
  \begin{align*}
    \tilde{U}^{\dagger} \tilde{U}
  =
    {\textstyle \frac{{\dim(\Sym^k(\cH))}}{|\cV|}} \sum_{\nu \in \cV} \proj{\nu^{\otimes k}}
  \end{align*} 
  and, consequently, 
  \begin{align*}  
    \bigl\|
      \tilde{U}^{\dagger} \tilde{U}
    -
      P_{\Sym^k(\cH)} 
    \bigr\| 
  \leq 
    \mu \ .
  \end{align*}
  In particular, $\tilde{U}$ can be made arbitrarily close to the isometry
  \begin{align*}
    U := \tilde{U} (\tilde{U}^\dagger \tilde{U})^{-\frac{1}{2}}
  \end{align*}
  on $\Sym^k(\cH)$, i.e., for any $\mu' > 0$ there exists a finite set
  $\cV$ such that
  \begin{align*}
    \bigl\| \tilde{U} - U \| \leq \mu' \ .
  \end{align*}
  It thus remains to be shown that inequality~\eqref{eq:fidbound}
  holds for $\tilde{U}$ (because it then also holds for the isometry
  $U$, provided $\mu'$ is sufficiently small). 

  By the definition of $\tilde{U}$, the vector $(\tilde{U} \otimes
  \id^{\otimes k+n})\Phi$ can be written as
  \begin{align*}
    (\tilde{U} \otimes \id^{\otimes k+n})\Phi
  =
    {\textstyle \sqrt{\frac{1}{|\cV|}}} \sum_{\nu \in \cV} f_\nu \otimes \Phi_\nu
  \end{align*}
  where 
  \begin{align*}
    \Phi_\nu :={\textstyle \sqrt{\dim(\Sym^k(\cH))}} (\bra{\nu}^{\otimes k} \otimes \id_{k+n})\Phi \in \cH^{\otimes k+n} \ .
  \end{align*}
  We now define the vector $\hat{\Phi}$ by choosing each
  $\hat{\Phi}_\nu$ of the sum~\eqref{eq:Phidef} as the projection of
  $\Phi_\nu$ onto the subspace $\Sym^{k+n}(\cH, \nu^{\otimes n})$,
  \begin{align*}
    \hat{\Phi}_\nu := P_{\Sym^{k+n}(\cH, \nu^{\otimes n})} \Phi_\nu \ .
  \end{align*}
  Note that the length of the resulting vector $\hat{\Phi}$ is
  generally smaller than $1$. However, the statement for unit vectors
  can be obtained by normalizing $\hat{\Phi}$ (because the
  normalization can only increase the overlap).

  Condition~\eqref{eq:fidbound} (with $U$ replaced by $\tilde{U}$) can
  now be rewritten as
  \begin{align}
    {\textstyle \frac{1}{|\cV|}} \sum_{\nu \in \cV} \bra{\Phi_{\nu}}  P_{\Sym^{k+n}(\cH, \nu^{\otimes n})} \ket{\Phi_{\nu}}
  >
     1 - k^d e^{- \frac{k (k+1)}{2 k + n}} \ ,
  \end{align}
  or, equivalently, as 
  \begin{align} \label{eq:phibound}
    {\textstyle \frac{1}{|\cV|}} \sum_{\nu \in \cV} \bra{\Phi_{\nu}}  P^\perp_{\Sym^{k+n}(\cH, \nu^{\otimes n})} \ket{\Phi_{\nu}}
  <
     k^d e^{- \frac{k (k+1)}{2 k + n}}  \ ,
  \end{align}
  because $\frac{1}{|\cV|} \sum_{\nu \in \cV} \spr{\Phi_\nu}{\Phi_\nu}
  \geq 1-\mu$, for any $\mu > 0$.

  A straightforward calculation (cf.\ Eq.~(4.12) of~\cite{Renner05} or
  the supplementary material of~\cite{Renner07} for a similar but more
  detailed argument) shows that, for any vector $\Psi \in
  \Sym^{2k+n}(\cH)$,
  \begin{multline*}
    \bra{\Psi} (\proj{\nu})^{\otimes k} \otimes P^\perp_{\Sym^{k+n}(\cH, \nu^{\otimes n})} \ket{\Psi}
  \leq
    \frac{\binom{k+n}{k+1}}{\binom{2k+n}{k+1}} \\
  \leq
    \Bigl(\frac{k+n-1}{2k + n}\Bigr)^k
  <
    e^{-\frac{k(k+1)}{2k + n}}  \ .
  \end{multline*}
  Applying this bound to the individual terms in the
  sum~\eqref{eq:phibound} gives
  \begin{align*}
    \bra{\Phi_{\nu}}  P^\perp_{\Sym^{k+n}(\cH, \nu^{\otimes n})} \ket{\Phi_{\nu}}
  & <
     {\textstyle \dim(\Sym^k(\cH))} e^{- \frac{k (k+1)}{2 k + n}} \\
  & \leq
    k^d e^{- \frac{k (k+1)}{2 k + n}} \ .
  \end{align*}
  This implies~\eqref{eq:phibound} and thus concludes the proof.
\end{proof}

Based on Lemma~\ref{lem:dF}, we now derive a de Finetti-type theorem
that applies to states $\rho^{4k+n}$ on the symmetric subspace
$\Sym^{4k+n}(\cH, \cHb^{\otimes 3k+n})$, where $\cHb$ is a
finite-dimensional subspace of a possibly infinite-dimensional Hilbert
space $\cH$. The claim is that, when tracing out the first $2 k$
subsystems, the resulting state $\rho^{2k + n} =
\tr_{2k}(\rho^{4k+n})$ is close to a convex combination of
$\binom{2k+n}{n}$-i.i.d.\ states $\hat{\rho}_{\nu}^{2k+n}$. Here,
closeness is measured in terms of the fidelity $F(\cdot, \cdot)$.

\begin{theorem} \label{thm:dFext} Let $\cHb$ be a $d$-dimensional
  subspace of a Hilbert space $\cH$, let $n,k \in \mathbb{N}$, and let
  $\rho^{4k+n}$ be a density operator on $\Sym^{4k+n}(\cH,
  \cHb^{\otimes 3k + n})$.  Then there exists a probability
  distribution $p_\nu$ on a finite set $\cV$ of unit vectors $\nu \in
  \bar{\cH}$ and a family $\{\hat{\rho}^{2k+n}_\nu\}_{\nu \in \cV}$ of
  density operators on $\Sym^{2k+n}(\cH, \nu^{\otimes n})$ such that
  \begin{align} \label{eq:dFextclaim}
    F\bigl(\rho^{2k + n}, \sum_{\nu \in \cV} p_\nu \hat{\rho}^{2k+n}_\nu\bigr) >  1 - k^{d} e^{- \frac{k(k+1)}{4k + n}}\ .
  \end{align}  
\end{theorem}

\begin{proof}
  It suffices to prove the claim for $\rho^{4k+n}$ pure; the statement
  for general density operators follows by the joint concavity of the
  fidelity (see, e.g., Chapter~9 of~\cite{NieChu00}).  Let thus $\Psi
  \in \Sym^{4k+n}(\cH, \cHb^{\otimes 3k+n})$. The idea is to write
  $\Psi$ as a superposition of vectors $\Psi_{\mathbf{j},\mathbf{j'}}$
  which have at least $2k+n$ subsystems contained in $\cHb$
  (see~\eqref{eq:superpos} and~\eqref{eq:superpostwo} below) so that
  we can apply Lemma~\ref{lem:dF} to each of them individually.

  Consider the decomposition of $P^{2k+n}_{\cHb^{\otimes k+n}}$ according
  to~\eqref{eq:Pkndef}, i.e., 
  \begin{align} \label{eq:Pknrep}
    P^{2k+n}_{\cHb^{\otimes k+n}}
  =
    \sum_{\substack{\mathbf{b} \in \{0,1\}^{2k+n}\\f_{\mathbf{b}} \leq \frac{k}{2k+n}}} P_{b_1} \otimes \cdots \otimes P_{b_{2k+n}} 
  \end{align}
  with $P_0 = P_{\cHb}$ and $P_1 = P_{\cHb^\perp}$. Furthermore, let
  $\{e_j\}_{j \in J}$ be a common eigenbasis of the projectors
  $P_{\cHb}$ and $P_{\cHb^\perp}$, let $J_0 := J \cup \{0\}$ (assuming
  that $0 \notin J$), and define the projectors $Q_j$, for $j \in
  J_0$, by
  \begin{align*}
    Q_j = \begin{cases} P_{\cHb} & \text{if $j = 0$} \\ \proj{e_j} & \text{if $j \in J$.} \end{cases}
  \end{align*}
  Then, starting from~\eqref{eq:Pknrep}, it is easy to construct a
  decomposition of $P^{2k+n}_{\cHb^{\otimes k+n}}$ into mutually
  orthogonal projectors $Q_{\mathbf{j'}} = Q_{j'_{1}} \otimes \cdots
  \otimes Q_{j'_{2k +n}}$, for $\mathbf{j'} = (j'_1, \ldots,
  j'_{2k+n})$,
  \begin{align} \label{eq:Pknone}
    P^{2k+n}_{\cHb^{\otimes k+n}} 
  = 
    \sum_{\mathbf{j'} \in \mathbf{J}^{2k+n}_{k+n}}   Q_{\mathbf{j'}} \ ,
  \end{align}
  where $\mathbf{J}^{2k+n}_{k+n}$ is a subset of $J_0^{2k+n}$
  containing only tuples $\mathbf{j'}$ with exactly $k+n$ indices
  $\tau$ such that $j'_\tau = 0$. Similarly, we can decompose
  $P^{2k}_{\cHb^{\otimes k}}$ in projectors $Q_{\mathbf{j}} = Q_{j_1}
  \otimes \cdots \otimes Q_{j_{2k}}$, for $\mathbf{j} = (j_1, \ldots,
  j_{2k})$,
  \begin{align} \label{eq:Pknone}
    P^{2k}_{\cHb^{\otimes k}}
  = 
    \sum_{\mathbf{j} \in \mathbf{J}^{2k}_{k}} Q_{j_1} \otimes \cdots \otimes Q_{j_{2k}}  \ ,
  \end{align}
  where $\mathbf{J}^{2k}_k$ only consists of tuples $\mathbf{j} \in
  J_0^{2k}$ with exactly $k$ indices $\tau$ such that $j_\tau = 0$.

  By definition, $\Sym^{4k+n}(\cH, \cHb^{\otimes 3k+n})$ is contained
  in the support of $P^{\otimes 4 k + n}_{\cHb^{\otimes 3k+n}}$, which
  is itself contained in the support of $P^{2k}_{\cHb^{\otimes k}}
  \otimes P^{2k + n}_{\cHb^{\otimes k+n}}$. Hence, any vector $\Psi
  \in \Sym^{4k+n}(\cH, \cHb^{\otimes 3k+n})$ can be written as a
  superposition
  \begin{align} \label{eq:superpos}
    \Psi = \sum_{\mathbf{j} \in \mathbf{J}^{2k}_k} \Psi_{\mathbf{j}}
  \end{align}
  where, for any $\mathbf{j} \in \mathbf{J}^{2k}_k$,
  \begin{align} \label{eq:superpostwo}
    \Psi_{\mathbf{j}} 
  =
    (Q_{\mathbf{j}} \otimes \id^{\otimes 2k+n}) \Psi
  = 
    \sum_{\mathbf{j'} \in \mathbf{J}^{2k + n}_{k+n}} \Psi_{\mathbf{j}, \mathbf{j'}} \, 
  \end{align}
  with $ \Psi_{\mathbf{j}, \mathbf{j'}} = Q_{\mathbf{j}} \otimes
  Q_{\mathbf{j'}} \Psi $. In particular, because of the orthogonality
  of the projectors $Q_{\mathbf{j}}$, we can write $\rho^{2k + n}$ as
  a convex combination
  \begin{align*}
    \rho^{2k + n}
  =
    \tr_{2k}(\proj{\Psi})
  & =
    \sum_{\mathbf{j} \in \mathbf{J}^{2k}_k} \tr_{2k}(\proj{\Psi_{\mathbf{j}}})  \\
  & =
    \sum_{\mathbf{j} \in \mathbf{J}^{2k}_k} p_{\mathbf{j}} \rho^{2k+n}_{\mathbf{j}} 
  \end{align*}
  with probabilities $p_{\mathbf{j}} = \tr(\proj{\Psi_{\mathbf{j}}})$
  and density operators
  \begin{align*}
    \rho^{2 k + n}_{\mathbf{j}} = \tr_{2k}(\proj{\tilde{\Psi}_{\mathbf{j}}}) \ ,
  \end{align*}
  where $\tilde{\Psi}_{\mathbf{j}}$ is a unit vector parallel to
  $\Psi_{\mathbf{j}}$. Because of the joint concavity of the fidelity,
  it is thus sufficient to show that~\eqref{eq:dFextclaim} holds for
  all density operators $\rho^{2k+n}_{\mathbf{j}}$.

  Let thus $\mathbf{j} \in \mathbf{J}^{2k}_k$ be fixed and let, for
  any $\mathbf{j'} \in \mathbf{J}^{2k +n}_{k+n}$,
  $\tilde{\Psi}_{\mathbf{j},\mathbf{j'}}$ be a normalization of
  $\Psi_{\mathbf{j}, \mathbf{j'}}$. We then have
  \begin{align*}
    \tilde{\Psi}_{\mathbf{j}} 
  =
    \sum_{\mathbf{j'} \in \mathbf{J}^{2k + n}_{k+n}} \alpha_{\mathbf{j},\mathbf{j'}} \tilde{\Psi}_{\mathbf{j},\mathbf{j'}} 
  \end{align*}
  where $\alpha_{\mathbf{j},\mathbf{j'}}$ are coefficients satisfying
  $\sum_{\mathbf{j'}} |\alpha_{\mathbf{j},\mathbf{j'}}|^2 = 1$.  We
  now apply Lemma~\ref{lem:dF} to each of the vectors
  $\tilde{\Psi}_{\mathbf{j},\mathbf{j'}}$ in the sum individually. For
  this, assume without loss of generality that $\mathbf{j} = (j_1,
  \ldots, j_k, 0, \ldots, 0)$ and $\mathbf{j'} = (0, \ldots, 0,
  j'_{k+n+1}, \cdots, j'_{2k + n})$ with $j_{1}, \ldots, j_{k},
  j'_{k+n+1}, \ldots, j'_{2k+n} \in J$ (this form can always be
  obtained by an appropriate reordering of the subsystems). The vector
  $\tilde{\Psi}_{\mathbf{j},\mathbf{j'}}$ can then be written as
  \begin{align*}
    \tilde{\Psi}_{\mathbf{j},\mathbf{j'}}
  =
    e_{j_1} \otimes \cdots \otimes e_{j_k} \otimes \Phi_{\mathbf{j},\mathbf{j'}} \otimes e_{j'_{k+n+1}} \otimes \cdots \otimes e_{j'_{2k+n}}
  \end{align*}
  where $\Phi_{\mathbf{j}, \mathbf{j'}} \in \Sym^{2k +
    n}(\cHb)$. According to Lemma~\ref{lem:dF}, there exists a vector
  $\hat{\Phi}_{\mathbf{j},\mathbf{j'}}$ of the form
  \begin{align*}
    \hat{\Phi}_{\mathbf{j},\mathbf{j'}}
  =
    {\textstyle \sqrt{\frac{1}{|\cV|}}} \sum_{\nu \in \cV} f_\nu \otimes \hat{\Phi}_{\mathbf{j}, \mathbf{j'}, \nu} \in {\textstyle \cH' \otimes \Sym^{k+n}(\cHb) }
  \end{align*}
  with $\hat{\Phi}_{\mathbf{j}, \mathbf{j'}, \nu} \in \Sym^{k+n}(\cHb,
  \nu^{\otimes n})$ such that
  \begin{align*}
    \spr{\hat{\Phi}_{\mathbf{j},\mathbf{j'}}}{(U \otimes \id^{\otimes k + n})\Phi_{\mathbf{j}, \mathbf{j'}}}
  >
    1 - k^d e^{-\frac{k(k+1)}{2k+n}} \ ,
  \end{align*}
  where $U$ is some fixed isometry (independent of
  $\mathbf{j'}$). With the definition
  \begin{align*}
   \hat{\Psi}_{\mathbf{j},\mathbf{j'}}
  =
    e_{j_1} \otimes \cdots \otimes e_{j_k} \otimes \hat{\Phi}_{\mathbf{j},\mathbf{j'}} \otimes e_{j'_{k+n+1}} \otimes \cdots \otimes e_{j'_{2k+n}}  
  \end{align*}
  this immediately implies 
  \begin{align*}
    \spr{\hat{\Psi}_{\mathbf{j},\mathbf{j'}}}{(\id^{\otimes k} \otimes U \otimes \id^{\otimes 2k + n})\tilde{\Psi}_{\mathbf{j}, \mathbf{j'}}}
  >
    1 - k^d e^{-\frac{k(k+1)}{2k+n}}  \ .
  \end{align*}
  Consider now the vector
  \begin{align*}
    \hat{\Psi}_{\mathbf{j}}
  :=
    \sum_{\mathbf{j'} \in \mathbf{J}^{2k+n}_{k+n}} \alpha_{\mathbf{j}, \mathbf{j'}} \hat{\Psi}_{\mathbf{j}, \mathbf{j'}} \ .
  \end{align*}
  Since, for any two distinct $\mathbf{j'}, \mathbf{j''} \in
  \mathbf{J}^{2k+n}_{k+n}$, the projectors $Q_{\mathbf{j'}}$ and
  $Q_{\mathbf{j''}}$ are mutually orthogonal by definition, we have
  \begin{multline*}
    \spr{\hat{\Psi}_{\mathbf{j}, \mathbf{j'}}}{(\id^{\otimes k}
    \otimes U \otimes \id^{\otimes 2k + n})\tilde{\Psi}_{\mathbf{j},
      \mathbf{j''}}}  \\
  =
    \spr{\hat{\Psi}_{\mathbf{j}, \mathbf{j'}}}{(\id^{\otimes k}
    \otimes U \otimes Q_{\mathbf{j'}} Q_{\mathbf{j''}} )\tilde{\Psi}_{\mathbf{j},
      \mathbf{j''}}} 
  = 
    0 \ .
  \end{multline*}
  Combining this with the above, we find
  \begin{multline*}
    \spr{\hat{\Psi}_\mathbf{j}}{(\id^{\otimes k} \otimes U \otimes \id^{\otimes 2k + n})\tilde{\Psi}_{\mathbf{j}}}  \\
  =
    \sum_{\mathbf{j'}} |\alpha_{\mathbf{j}, \mathbf{j'}}|^2     \spr{\hat{\Psi}_{\mathbf{j},\mathbf{j'}}}{(\id^{\otimes k} \otimes U \otimes \id^{\otimes 2k + n})\tilde{\Psi}_{\mathbf{j}, \mathbf{j'}}} \\
  >
    1 - k^d e^{-\frac{k(k+1)}{2k+n}}  \ .
  \end{multline*}

  This inequality can be rewritten in terms of the fidelity, which is
  simply the absolute value of the scalar product.  Together with the
  fact that tracing out subsystems can only increase the fidelity, we
  obtain
  \begin{align*}
    F\bigl(\tr_{2k}(\proj{\tilde{\Psi}_{\mathbf{j}}}), \tr_{\cH',k}(\proj{\hat{\Psi}_{\mathbf{j}}})\bigr)
  > 
    1 - k^d e^{-\frac{k(k+1)}{2k+n}} \ .
  \end{align*}
  Furthermore, because the density operator
  $\tr_{2k}(\proj{\tilde{\Psi}_{\mathbf{j}}})$ is contained in the
  symmetric subspace $\Sym^{2k+n}(\cH)$, we can insert a projection
  onto this subspace without changing the fidelity, i.e., 
  \begin{align*}
     F\bigl(\tr_{2k}(\proj{\tilde{\Psi}_{\mathbf{j}}}), \hat{\rho}^{2k+n}_{\mathbf{j}}\bigr)
  > 
    1 - k^d e^{-\frac{k(k+1)}{2k+n}} \ ,
  \end{align*}
  where 
  \begin{align*}
    \hat{\rho}^{2k+n}_{\mathbf{j}} := P_{\Sym^{2k+n}(\cH)}
  \tr_{k,\cH'}(\proj{\hat{\Psi}_{\mathbf{j}}}) P_{\Sym^{2k+n}(\cH)} \ .
  \end{align*} 
  
  It remains to verify that the density operator
  $\hat{\rho}^{2k+n}_{\mathbf{j}}$ is of the desired form
  \begin{align} \label{eq:convnu}
     \hat{\rho}_{\mathbf{j}}^{2k+n}
  =
    \sum_{\nu \in \cV} p_{\nu} \hat{\rho}_{\mathbf{j},\nu}^{2k +n}
  \end{align} 
  for some appropriately chosen probabilities $p_{\nu}$ and for
  $\hat{\rho}_{\mathbf{j},\nu}^{2k +n}$ contained in the subspace
  $\Sym^{2k+n}(\cH, \nu^{\otimes n})$. For this, we define
  \begin{align*}
    \hat{\rho}_{\mathbf{j},\nu}^{2k +n} := P_{\Sym^{2k+n}(\cH)} P_{\hat{\Psi}_{\mathbf{j}, \nu}} P_{\Sym^{2k+n}(\cH)}
  \end{align*}
  where $P_{\hat{\Psi}_{\mathbf{j}, \nu}}$ denotes the projector onto the
  vector
  \begin{align*}
    \hat{\Psi}_{\mathbf{j}, \nu} 
  :=
    \bra{e_{j_1} \otimes \cdots \otimes e_{j_k} \otimes f_{\nu}}{\hat{\Psi}_{\mathbf{j}}} \ .
  \end{align*}
  Identity~\eqref{eq:convnu} then follows from the orthogonality of
  the vectors $f_{\nu}$. Furthermore, by the definition of
  $\hat{\Psi}_{\mathbf{j}}$ and using the fact that the vectors
  $\hat{\Phi}_{\mathbf{j}, \mathbf{j'},\nu}$, for any fixed $\nu$ and
  arbitrary $\mathbf{j'}$, are contained in the support of
  $P^{k+n}_{\nu^{\otimes n}}$, one can readily verify that the vector
  $\hat{\Psi}_{\mathbf{j}, \nu}$ is contained in the support of
  $P^{2k+n}_{\nu^{\otimes n}}$. Consequently,
  $\hat{\rho}_{\mathbf{j},\nu}^{2k +n}$ lies in the subspace
  $\Sym^{2k+n}(\cH, \nu^{\otimes n})$.  
\end{proof}

\subsection{Properties of almost i.i.d.\ states} \label{sec:properties}

Theorem~\ref{thm:dFext} gives an approximation of permutation
invariant states in terms of almost i.i.d.\ states $\rho_{\nu}$. The
significance of this approximation comes from the fact that such
states are relatively easy to handle. In particular, their properties
very much resemble the properties of (perfect) i.i.d.\
states~\cite{Renner05,Renner07}. For example, the entropy of an almost
i.i.d.\ state $\rho_{\nu}$ is well approximated by the entropy of the
corresponding perfect i.i.d.\ state .

Of particular interest for information-theoretic applications is the
\emph{smooth min-entropy}~\cite{RenKoe05,Renner05}. Let $\rho_{X B}$
be a density operator on $\cH_X \otimes \cH_B$ which is classical on
$\cH_X$, i.e.,
\begin{align*}
  \rho_{X B} = \sum_{x \in \cX} p_x \proj{e_x} \otimes \rho^x_B \ ,
\end{align*}
for some orthonormal basis $\{e_x\}_{x \in \cX}$ of $\cH_X$,
probabilities $p_x$, and density operators $\rho^x_B$ on
$\cH_B$. Then, for any $\eps \geq 0$, the \emph{$\eps$-smooth entropy
  of $X$ given $B$}, denoted $H_{\min}^\eps(X|B)_\rho$, corresponds to
the amount of uniform randomness (relative to $B$) that can be
extracted from $X$ by two-universal hashing~\cite{WegCar81}. (The
smoothness parameter $\eps$ quantifies the quality of the resulting
randomness in terms of their distance to a random variable which is
perfectly uniform and independent of $B$).

For an i.i.d.\ state $\rho_{X B}^{\otimes N}$, the smooth min-entropy
$H_{\min}^\eps$ is asymptotically (for large $N$) equal to the von
Neumann entropy $S$, i.e.,
\begin{align*}
  \frac{1}{N} H_{\min}^\eps(X^N|B^N)_{\rho^{\otimes N}} 
& \approx
  \frac{1}{N} \bigl( S(\rho_{X B}^{\otimes N}) - S(\rho_B^{\otimes N}) \bigr) \\
& =
  S(\rho_{X B}) - S(\rho_B)
= S(X|B) \ .
\end{align*}

The following theorem from~\cite{Renner05} extends (one direction of)
this relation to almost i.i.d.\ states. 

\begin{theorem} \label{thm:entropy} Let $\rho_{X^{k+n} B^{k+n}}$ be a
  density operator on $(\cH_X \otimes \cH_B)^{\otimes k+n}$ which is
  classical on the subsystems $\cH_X$ and let $\eps > 0$. If there
  exists a purification of $\rho_{X^{k+n} B^{k+n}}$ in
  $\Sym^{k+n}({\cH_X \otimes \cH_B \otimes \cH_R}, \nu^{\otimes n})$,
  for some $\nu \in {\cH_X \otimes \cH_B \otimes \cH_R}$, then
  \begin{align*}
    \frac{1}{n} H_{\min}^\eps(X^{k+n}|B^{k+n})_{\rho^{k+n}}
  \geq
    S(\sigma_{X B}) - S(\sigma_B) - \delta \ ,
  \end{align*}
  where $\sigma_{X B} := \tr_{R}(\proj{\nu})$, 
  \begin{align*} 
    {\textstyle \delta := 5 (\ln(\dim\cH_X)+1) \sqrt{\frac{2 \ln(4/\eps)}{k+n} +
      h(\frac{k}{k+n})} }\ ,
  \end{align*}
  and $h(p) \equiv -p \ln p - (1-p) \ln (1-p)$.
\end{theorem}

Note that the statement depends on the dimension of $\cH_X$, but is
independent of the dimension of $\cH_B$.

\section{Implications} 

\subsection{Putting things together} \label{sec:combination}

The aim of this section is to demonstrate how the technical statements
of Section~\ref{sec:tech} can be combined to give our main claim,
namely that any permutation invariant state $\rho^N$ on $\cH^{\otimes
  N}$ is approximated by a mixture of states with almost i.i.d.\
structure, provided the outcomes of certain measurements on a (small)
sample of the subsystems lie in a given range. To illustrate this, we
assume for concreteness that $\cH = {\rm L}^2(\mathbb{R})$ and that
measurements on $k$ subsystems are carried out with respect to two
canonical observables $X$ and $Y$, each chosen with probability
$\frac{1}{2}$. (According to Remark~\ref{rem:lowdimbound}, the
argument below can easily be extended to more general measurements.)
Furthermore, we assume that $N = m^4$ and $k=m^3$, for some $m \in
\mathbb{N}$.

Let $d = m^{\frac{3}{2}}$ and let $\cHb$ be the support of $P^{X^2 +
  Y^2 \leq n_0}$ for some $n_0 \in \mathbb{N}$ such that $12 \ln (7 m)
\leq n_0 \leq d$.  We first apply Lemma~\ref{lem:lowdimbound} to infer
that, if all $k$ measurement outcomes $z_1, \ldots, z_k$ satisfy
$z_i^2 \leq \frac{n_o}{2}$ then the state $\rho^{(m-1)k}$ on the
remaining $(m-1)k$ subsystems is almost certainly contained in the
support of $P^{(m-1)k}_{\cHb^{\otimes (m-2)k}}$. Hence, according to
Lemma~\ref{lem:pur}, there exists a purification $\bar{\rho}^{(m-1)k}$
of $\rho^{(m-1)k}$ on $\Sym^{(m-1)k}({\cH \otimes \cH}, {(\cHb \otimes
  \cHb)}^{\otimes (m-3)k})$.  Theorem~\ref{thm:dFext} now provides an
approximation of the reduced state $\bar{\rho}^{(m-5) k}$ in terms of
a mixture of almost i.i.d.\ states $\hat{\rho}^{(m-5) k}_{\nu}$,
parametrized by $\nu \in \cHb \otimes \cHb$. More precisely, each
density operator $\hat{\rho}^{(m-5) k}_{\nu}$ is contained in
$\Sym^{(m-5) k}(\cH \otimes \cH, \nu^{\otimes ((m-9) k})$, and their
convex combination is exponentially (in $m$) close to
$\bar{\rho}^{(m-5) k}$. In particular, by taking the trace over the
purifying systems, we conclude that the reduced state $\rho^{(1-\mu)
  N}$ is approximated by a mixture of states that have i.i.d.\
structure on $(1- \mu - \mu') N$ subsystems, where $\mu = 5
N^{-\frac{1}{4}}$ and $\mu' = 4 N^{-\frac{1}{4}}$.

\subsection{Application to QKD} \label{sec:application}

A main application of de Finetti's representation theorem is in the
area of quantum information theory. As explained in the introduction,
the theorem can be employed for the analysis of schemes involving a
large number of information carriers, whose joint state may be
difficult to describe in general. A typical and practically relevant
example is QKD, where the challenge is to find security proofs that
take into account all possible attacks of an adversary.

Most QKD protocols can be subdivided into two parts. In the first
part, also known as \emph{distribution phase}, the two legitimate
parties, Alice and Bob, use an (insecure) quantum communication
channel in order to distribute correlated information. (Alternatively,
in an entanglement-based scheme~\cite{Ekert91}, Alice and Bob receive
this correlated information as an input from an external source, which
may be controlled by an adversary.) In the second part, the
\emph{distillation phase}, Alice and Bob process this information to
extract a pair of secret keys. This process usually only involves
classical communication (over an authentic channel).

The analysis based on de Finetti's theorem sketched below applies to a
large class of QKD schemes, which includes almost all protocols
proposed in the literature~\footnote{Among the few exceptions are the
  \emph{Differential Phase Shift (DPS)}~\cite{InWaYa02} and the
  \emph{Coherent One-Way (COW) Protocol}~\cite{SBGSZ05}. Both rely on
  measurements involving two subsequent signals at the same time, so
  that the order in which the signals are received is relevant.}.
More concretely, the following conditions must hold.
\begin{enumerate}
\item \label{pr:sym} We assume that the information held by Alice and
  Bob after the distribution phase consists of $N$ parts, for some
  sufficiently large $N$. The protocol should be invariant under
  permutations of these parts. This requirement is usually satisfied
  because each of the $N$ signals is prepared, sent, and received
  independently of the other signals.
\item \label{pr:dist} In the last step of the distillation phase, the
  final key is computed in a classical post-processing procedure
  consisting of information reconciliation (error correction) and
  privacy amplification by two-universal hashing~\cite{WegCar81}. As
  yet, no alternative method for distilling the final key is known, so
  this criterion is not restrictive~\footnote{Note that the
    distillation phase may involve other steps such as \emph{sifting}
    or \emph{advantage distillation}, but these need to be carried out
    on single signals (or small blocks of signals) independently.}.
\item \label{pr:test} The protocol must perform a measurement $\cM =
  \{M_z\}_{z \in \cZ}$ on a sample of the received signals and only
  continue if all outcomes $z$ are contained in a given set $\cZb
  \subset \cZ$ that allows to conclude that the dimension of the
  relevant Hilbert space $\cHb$ is finite (cf.\
  Remark~\ref{rem:lowdimbound}). Note that this requirement is
  trivial if the signal space already has small dimension.

  A concrete example in $\cH = {\rm L}^2(\mathbb{R})$ are measurements
  $\cM$ with respect to two canonical observables $X$ and $Y$, each of
  them chosen with probability $\frac{1}{2}$. The set $\cZb$ can then
  be defined as the set of all outcomes $z$ such that $z^2 \leq
  \frac{n_0}{2}$ and $\cHb$ is the space spanned by the eigenvectors
  of $X^2 + Y^2$ corresponding to eigenvalues larger than $n_0$, for
  some appropriately chosen $n_0$ (see Lemma~\ref{lem:Nconstraint} and
  Lemma~\ref{lem:lowdimbound}).
\end{enumerate}

According to Property~\ref{pr:sym}, if a key distilled from $N$
signals in state $\rho^{N}$ is secure then the same is true for the
key distilled from a permuted state $\pi \rho^N \pi^\dagger$, for any
permutation $\pi \in S_N$. We can thus assume without loss of
generality that the $N$ signals are permuted at random and, hence,
their state $\rho^N$ is permutation invariant. Now, according to the
argument in Section~\ref{sec:combination} and using
Property~\ref{pr:test}, we conclude that $\rho^{N'}$, for some $N'
\approx N$, is approximated by a mixture of almost i.i.d.\ states
$\rho_\nu$ (see previous section for explicit parameters). Finally, we
use Property~\ref{pr:dist}, which implies that the only relevant
quantity is the smooth min-entropy~\cite{RenKoe05} of the measured
data $X^N$ conditioned on the adversary's information $E^N$
(see~\cite{Renner05} for a detailed argument). By
Theorem~\ref{thm:entropy}, the smooth min-entropy of almost i.i.d.\
states is approximated by the corresponding entropy of i.i.d.\
states~\footnote{Theorem~\ref{thm:entropy} can be applied because
  $X^N$ usually is a sequence of digitally represented values, so that
  $\cH_X$ has finite dimension.} Hence, we can without loss of
generality assume that $\rho^N$ is an i.i.d.\ state, which could
equivalently be the result of a collective attack. Summarizing, we
have thus proved that any QKD protocol satisfying the above three
conditions is secure against general attacks whenever it is secure
against collective attacks.

\section{Conclusions}

We have shown that permutation invariant states on large $N$-partite
systems are approximated by a convex combination of almost i.i.d.\
states, provided measurements on a few subsystems with respect to
certain observables only give bounded values.  In particular, under
this condition, a permutation invariant state can be considered equal
to an unknown i.i.d.\ state, except an arbitrarily small fraction of
the subsystems. This has various implications. Of particular interest
to experimental physics is that state tomography can be employed
without the need for i.i.d.\ assumptions, as discussed
in~\cite{Renner07} for the special case of low-dimensional systems.

Applied to quantum cryptography, our result enables full security
proofs for QKD schemes in the (practically relevant case) where the
dimension of the signal space may be unbounded. This is an intrinsic
property of continuous variable protocols, but the necessity of taking
into account infinite-dimensional systems may also arise in the
analysis of discrete variable schemes, for instance when they are
implemented using weak coherent pulses (see, e.g., \cite{BLMS00}). The
security of these schemes has been investigated intensively, but most
proofs are only valid under the assumption of collective attacks (see
Introduction for references and~\footnote{A remarkable exception are
  QKD schemes using Gaussian states, for which security against
  general attacks can also be proved using a result on the extremality
  of Gaussian states~\cite{WoGiCi06}, as shown in work done in
  parallel to ours~\cite{LKGC08}.}). The de Finetti representation
theorem derived here allows to drop this assumption, implying that
security holds against all possible attacks. The main requirement is
that certain tests are carried out on a sample of the transmitted
signals. For continuous variable protocols with signal states on $\cH
= {\rm L}^2(\mathbb{R})$, one possibility is to check that
measurements with respect to two canonical observables only result in
small outcomes. By modifying Lemma~\ref{lem:Nconstraint}, one may
replace this requirement by a criterion based on alternative
measurable quantities such as the photon number~\cite{Zhao08}.

\section{Acknowledgments}

IC acknowledges support from the EU project COMPAS and Caixa
Manresa. RR received support from the EU project SECOQC. He also would
like to thank Johan {\AA}berg for interesting and informative
discussions.


\appendix

\section{Proof of Lemma~\ref{lem:probbound}}

The proof is based on the following lemma, which states that the
statistics obtained from the observation of $k$ out of $k+n$ binary
values $X_1, X_2, \ldots, X_{k+n}$ gives a good estimate for the
probability distribution of any of the remaining values, provided the
overall distribution is permutation invariant.

\begin{lemma} \label{lem:statist} Let $n \geq k$ and let $P_{X_1,
    \ldots, X_{k+n}}$ be a permutation invariant probability
  distribution over $\{0, 1\}^{k+n}$. Then
  \begin{align*}
    \Pr\bigl[
      | p_{| X_1 \cdots \ldots X_k} - f_{X_1 \cdots X_k} | \geq \delta
    \bigr] \leq 2 k^{\frac{3}{2}} e^{-k \delta^2} \ ,
  \end{align*}
  where, for any $\mathbf{x} = (x_1, \ldots, x_k)$, $p_{|\mathbf{x}}$
  denotes the probability that $X_{k+1} = 1$ conditioned on $(X_1,
  \ldots, X_k) = \mathbf{x}$.
\end{lemma}

\begin{proof}
  Let $\mathbf{X} = (X_1, \ldots, X_k)$.  We show that
  \begin{align} \label{eq:Expbound}
    \Expect\bigl[
      e^{k ( p_{|\mathbf{X}} - f_{\mathbf{X}} )^2}
    \bigr] \leq 2 k^{\frac{3}{2}} \ ,   
  \end{align}
  where $\Expect[\cdot]$ denotes the expectation value.  The claim
  then follows because, by Markov's inequality,
  \begin{multline*}
    \Pr\bigl[
      | p_{|\mathbf{X}} - f_{\mathbf{X}} | \geq \delta
    \bigr] 
  =
    \Pr\bigl[
      e^{k ( p_{|\mathbf{X}} - f_{\mathbf{X}} )^2} \geq e^{k \delta^2}
    \bigr] \\
  \leq 
    \Expect\bigl[
      e^{k ( p_{|\mathbf{X}} - f_{\mathbf{X}} )^2}
    \bigr] e^{-k \delta^2} 
  \leq
    2 k^{\frac{3}{2}} e^{-k \delta^2} \ .
  \end{multline*}

  To show~\eqref{eq:Expbound} we use the observation that, for any
  permutation invariant distribution $P_{Z_1 \cdots Z_k}$ of binary
  values, the distribution of any individual value $Z_i$ equals the
  expectation of the frequency distribution of the whole tuple $(Z_1,
  \ldots, Z_k)$, i.e., 
  \begin{align*}
    \Pr[Z_i=1] = \Expect[f_{Z_1 \cdots Z_k}] \ .
  \end{align*}
  In particular, we have for any $\mathbf{x} = (x_1, \ldots, x_k)$,
  \begin{align*}
    p_{|\mathbf{x}}
  =
    \Expect[f_{X_{k+1} \cdots X_{2k}}|\mathbf{X} = \mathbf{x}] \ .
  \end{align*}
  Using convexity of the function $x \mapsto e^{x^2}$, we get
  \begin{multline*}
    e^{k (p_{|\mathbf{x}} - f_{\mathbf{x}})^2}
  =
    e^{k \Expect[f_{X_{k+1} \cdots X_{2k}} - f_{\mathbf{X}}|\mathbf{X} = \mathbf{x}]^2} \\
  \leq
    \Expect\bigl[ e^{k (f_{X_{k+1} \cdots X_{2k}} - f_{\mathbf{X}})^2} \bigr|\mathbf{X} = \mathbf{x}]
  \end{multline*}
  and, hence,
  \begin{align*}
    \Expect\bigl[e^{k (p_{|\mathbf{X}} - f_{\mathbf{X}})^2} \bigr]
  \leq
    \Expect\bigl[e^{k (f_{X_{k+1} \cdots X_{2k}} - f_{\mathbf{X}})^2} \bigr] \ .
  \end{align*}
  It thus remains to be shown that
  \begin{align} \label{eq:exponentexpect}
     \Expect\bigl[e^{k (f_{X_{k+1} \cdots X_{2k}} - f_{X_1 \cdots X_k})^2} \bigr] 
  \leq 
     2 k^{\frac{3}{2}} \ .
  \end{align}
  for any permutation invariant distribution $P_{X_1 \cdots
    X_{2k}}$.

  Because any permutation invariant distribution can be written as a
  convex combination of permutation invariant distributions with fixed
  frequency distribution, we can without loss of generality assume
  that $f_{X_1 \cdots X_{2k}} = \frac{r}{2k}$ holds with certainty for
  any fixed $r \in \{0, \ldots, 2k\}$. The expectation value on the
  left hand side of~\eqref{eq:exponentexpect} is then given explicitly
  as
  \begin{multline} \label{eq:expectexpl}
    \Expect\bigl[e^{k (f_{X_{k+1} \cdots X_{2k}} - f_{X_1 \cdots X_k})^2} \bigr]  \\
  =
    \sum_{s=\max(0,r-k)}^{\min(r, k)} \frac{\binom{k}{s} \binom{k}{r-s}}{\binom{2k}{r}} e^{k (\frac{s}{k} - \frac{r-s}{k})^2} \\
  =
    \sum_{s=\max(0,r-k)}^{\min(r, k)} e^{- k (r_{k,r,s} - (\frac{s}{k} - \frac{r-s}{k})^2)}
  \end{multline}
  where
  \begin{align} \label{eq:rkrsdef}
    r_{k,r,s}
  & =
    \frac{1}{k} \ln \frac{\binom{2k}{r}}{\binom{k}{s} \binom{k}{r-s}} \ .
  \end{align}

  To bound the term $r_{k,r,s}$ we use an approximation of the
  binomial coefficient by Wozencraft and Reiffen~\cite{WozRei61} (see
  also Lemma~17.5.1 of~\cite{CovTho91})
  \begin{align*}
    \frac{e^{N h(p)}}{\sqrt{8 N g(p)}}
  \leq
    \binom{N}{p N} 
  \leq
    \frac{e^{N h(p)}}{\sqrt{\pi N g(p)}} \ ,
  \end{align*}
  where $h(p) \equiv - p \ln p - (1-p) \ln (1-p)$ is the binary
  entropy function (written with respect to the basis $e$) and where
  $g(p) \equiv p (1-p)$. The approximation holds for any $N \in
  \mathbb{N}$ and $0 < p < 1$ such that $p N \in \mathbb{N}$. Because
  $g(p) \leq \frac{1}{4}$ for any $p$, the first inequality implies
  \begin{align*}
    \binom{N}{p N} 
  \geq
    \frac{e^{N h(p)}}{\sqrt{2 N }} \ .
  \end{align*}
  Furthermore, since $g(p) \geq \frac{1}{2 N}$ for any $N > 1$ and
  $\frac{1}{N} \leq p \leq 1-\frac{1}{N}$, the second inequality
  implies the well known upper bound
  \begin{align*}
    \binom{N}{p N} 
  \leq
    e^{N h(p)} \ ,
  \end{align*}
  which also holds for $N=1$, $p=0$, and $p=1$. Inserting these bounds
  into~\eqref{eq:rkrsdef}, we find
  \begin{align} \label{eq:rkrsbound}
    r_{k,r,s}
  & \geq 
    2h(\frac{r}{2k}) - h(\frac{s}{k}) - h(\frac{r-s}{k})  
    - \frac{1}{2 k} \ln(4 k ) \ .
  \end{align}

  Using some standard analysis, one finds that
  \begin{align*}
    2 h(\frac{\alpha+\beta}{2} ) - h( \alpha) - h(\beta)
  \geq
    (\alpha - \beta)^2 
  \end{align*}
  for any $\alpha, \beta \in [0,1]$. Combining this
  with~\eqref{eq:rkrsbound} and inserting in~\eqref{eq:expectexpl}
  yields~\eqref{eq:exponentexpect} and thus concludes the proof.
\end{proof}

The following lemma is an immediate corollary of
Lemma~\ref{lem:statist}, applied to the sequence of values obtained
from measurements of a permutation invariant state.

\begin{lemma} \label{lem:quantstatist} Let $n \geq k$, let $\cU =
  \{U_0, U_1\}$ be a binary POVM on $\cH$, let $\rho^{k+n} \in
  \cS(\cH^{\otimes k+n})$ be permutation invariant, and let $(X_1,
  \ldots, X_k)$ be the outcome of the measurement $\cU^{\otimes k}$
  applied to the first $k$ subsystems of $\rho^{k+n}$. Then
  \begin{align*}
    \Pr\bigl[|\tr(U_1 \rho^1_{|X_1 \cdots X_k})  - f_{X_1 \cdots X_k}| \geq \delta \bigr]
  \leq
    2 k^{\frac{3}{2}} e^{-k\delta^2} \ ,
  \end{align*}
  where, for any $\mathbf{x} = (x_1, \ldots, x_k)$,
  $\rho^1_{|\mathbf{x}}$ is the reduced state on a single subsystem
  conditioned on the measurement outcome $(X_1, \ldots, X_k) =
  \mathbf{x}$.
\end{lemma}

We are now ready to prove Lemma~\ref{lem:probbound}

\begin{proof}[Proof of Lemma~\ref{lem:probbound}]
  Let $\mathbf{X}^k := (X_1, \ldots, X_k)$ and $\mathbf{X}^{n/2} :=
  (X_{k+1}, \ldots, X_{k + n/2})$ (where, for simplicity, we assume
  that $n$ is even). Applying Lemma~\ref{lem:quantstatist} to the
  density operator $\rho_{|\mathbf{X}^{n/2}}$ describing the state
  conditioned on the outcomes of measurement $\cV^{\otimes n/2}$
  applied to $n/2$ subsystems of $\rho^{k+n}$, we get
  \begin{align} \label{eq:probmboundone}
    \Pr\bigl[|\tr(U_1 \rho^1_{|\mathbf{X}^k\mathbf{X}^{n/2}}) - f_{\mathbf{X}^k} | \geq \delta \bigr] & \leq 2 k^{\frac{3}{2}} e^{-k \delta^2} \ .
  \end{align}
  Similarly and using $n/2 \geq k$ we find
  \begin{align} \label{eq:probmboundtwo}
   \Pr\bigl[|\tr(V_1 \rho^1_{|\mathbf{X}^k\mathbf{X}^{n/2}}) - f_{\mathbf{X}^{n/2}} | \geq \delta \bigr] & \leq 2 k^{\frac{3}{2}} e^{-k \delta^2} \ .
  \end{align}

  By the definition of the quantity $\gamma_{U_1 \to V_1}$, we have
  \begin{multline*}
    \Pr[f_{\mathbf{X}^{n/2}} > \gamma_{U_1 \to V_1}(f_{\mathbf{X}^{k}} + \delta) + \delta \bigr] \\
  =
    \Pr[\bigl[ \nexists \sigma : \, (\tr(U_1 \sigma) \leq
    f_{\mathbf{X}^k} + \delta) \wedge (\tr(V_1 \sigma) \geq f_{\mathbf{X}^{n/2}} - \delta) \bigr] \\
  \leq
    \Pr\bigl[ 
      (\tr(U_1 \rho^1_{|\mathbf{X}^k\mathbf{X}^{n/2}})  > f_{\mathbf{X}^k} + \delta) \vee
     (\tr(V_1 \rho^1_{|\mathbf{X}^k\mathbf{X}^{n/2}}) < f_{\mathbf{X}^{n/2}} - \delta)
    \bigr] \\
  \leq
    4 k^{\frac{3}{2}} e^{-k \delta^2} \ ,
  \end{multline*}
  where the last inequality follows from~\eqref{eq:probmboundone}
  and~\eqref{eq:probmboundtwo}, and the union bound.

  To conclude the proof, we use the observation that
  \begin{align*}
    f_{X_{k+1} \cdots X_{k+n}} = \frac{1}{2} f_{X_{k+1} \cdots X_{k+n/2}} + \frac{1}{2} f_{X_{k+n/2+1} \cdots X_{k+n}}
  \end{align*}  
  which implies
  \begin{multline*}
    \Pr[f_{X_{k+1} \cdots X_{k+n}} > \gamma_{U_1 \to V_1}(f_{\mathbf{X}^{k}} + \delta) + \delta \bigr] \\
  \leq
    \Pr[f_{X_{k+1} \cdots X_{k+n/2}} > \gamma_{U_1 \to V_1}(f_{\mathbf{X}^{k}} + \delta) + \delta \bigr] \\ \qquad
  +   \Pr[f_{X_{k+n/2+1} \cdots X_{k+n}} > \gamma_{U_1 \to V_1}(f_{\mathbf{X}^{k}} + \delta) + \delta \bigr] \\
  \leq
     8 k^{\frac{3}{2}} e^{-k \delta^2} \ .
  \end{multline*}
\end{proof}


\end{document}